\documentclass[12pt]{article}

\usepackage{amsmath,amssymb,amsfonts}
\usepackage{amsthm}
\usepackage{graphicx}
\usepackage{booktabs}
\usepackage{array}
\usepackage{tabularx}
\usepackage{url}
\usepackage{float}
\usepackage{placeins}
\usepackage{geometry}
\usepackage{hyperref}
\usepackage{adjustbox}

\newtheorem{definition}{Definition}
\newtheorem{lemma}{Lemma}
\newtheorem{theorem}{Theorem}
\newtheorem{corollary}{Corollary}
\newtheorem{example}{Example}
\newtheorem{remark}{Remark}

\newcommand{\Ecc}{\operatorname{ecc}}

\title{Multi-Orientation Edge-Minimum Repair for Non-Redundant Fault-Tolerant Broadcasting in Dense Eisenstein--Jacobi Networks}

\author{Bader A. Albader\\
\small Department of Computer Science, Faculty of Science, Kuwait University, Kuwait\\
\small \texttt{albader@cs.ku.edu.kw}}

\date{}

\begin{document}

\maketitle

\begin{abstract}
Dense Eisenstein--Jacobi (EJ) networks are degree-six algebraic interconnection networks whose finite quotient geometry is naturally represented by a hexagonal axial-coordinate ball. This paper studies non-redundant one-to-all broadcast repair in the dense EJ network generated by $\alpha=(t+1)+t\omega$, where $t$ is the network diameter. We propose EJ-MOEM, a multi-orientation edge-minimum repair method that evaluates a constant-size family of hexagonal broadcast-tree orientations, selects a fault-aware candidate, contracts the fault-pruned tree into healthy components, and reconnects these components using external component-crossing repair edges. The resulting structure is a rooted spanning tree of the healthy subgraph: every healthy node receives the message exactly once, no faulty node is used, and the original healthy tree components are preserved. We prove that, for a chosen orientation whose fault-pruned component graph is connected, exactly $c-1$ external repair edges are necessary and sufficient, where $c$ is the number of healthy components. We also prove a depth-certificate theorem for EJ coordinate-reduction trees: every one-fault placement admits a repair of depth at most $t+1$, and every two-fault placement admits a repair of depth at most $t+2$. The proof uses the three-strip representation of EJ hexagons, a sector-suffix attachment lemma, a non-adjacent-sector separation lemma, and a six-direction shielding classification for paired cuts. Extended validation includes exhaustive one- and two-fault enumeration for $t=2,\ldots,12,14,16,18$ (up to $N=1027$ and 525,825 two-fault placements at $t=18$), structured theorem-critical tests through $t=30$, and large random tests through $t=200$, all with 100\% success and no violation of the theorem.
\end{abstract}

\noindent\textbf{Keywords:} Eisenstein--Jacobi networks, hexagonal networks, Cayley graphs, fault-tolerant broadcasting, non-redundant communication, local repair, component repair, interconnection networks, edge-minimum repair, parallel communication.

\section{Introduction}
One-to-all broadcasting is a basic collective communication operation in parallel and distributed systems. In a non-redundant broadcast, each healthy processor receives the message exactly once. This avoids duplicate traffic, avoids unnecessary contention, and gives a clean correctness condition: the communication structure must be a rooted spanning tree of the healthy subgraph. The challenge is that a fault-free broadcast tree can be fragmented by a small number of faulty processors.

Algebraic interconnection networks provide compact symmetric graph models for such problems. Dense Eisenstein--Jacobi (EJ) networks are especially attractive because their degree is six, their distance balls are discrete hexagons, and their coordinate model is compatible with simple modular addressing. The dense EJ family considered in this paper has diameter $t$ and order
\begin{equation}
N=3t^2+3t+1.
\end{equation}
The natural fault-free source-centered broadcast tree reaches every node in at most $t$ steps. After faults occur, however, a fixed broadcast tree can split into multiple healthy components. A global breadth-first rebuild can restore reachability, but it may replace a large number of parent edges. The objective here is different: preserve the healthy portions of the damaged broadcast tree and add only the minimum number of external component-crossing repair edges needed to reconnect them.

This paper introduces \emph{EJ-MOEM}: multi-orientation edge-minimum repair for dense EJ broadcast trees. The method constructs a small family of deterministic EJ coordinate-reduction broadcast trees. For each orientation, the faulty nodes are deleted, the remaining forest is contracted into components, and a component-level repair is performed. The selected candidate is the valid repaired tree with best lexicographic score, primarily minimizing repaired depth and then repair count.

The contributions are as follows.
\begin{itemize}
\item We formulate non-redundant local broadcast repair for dense EJ networks under one and two processor faults.
\item We define a 15-orientation EJ broadcast-tree family based on six-direction coordinate reduction in the hexagonal axial ball.
\item We prove the component-minimum repair theorem: for a selected orientation with $c$ healthy components and connected component graph, exactly $c-1$ external component-crossing repair edges are necessary and sufficient.
\item We prove an EJ depth-certificate theorem: one fault is repaired within depth $t+1$, and two faults are repaired within depth $t+2$.
\item We give worked examples illustrating the coordinate model, orientation rules, one-fault repair, two-fault repair, and the component-contraction view.
\item We validate the proof-derived behavior by extended exhaustive, structured, and large random tests, including exhaustive two-fault enumeration up to $N=1027$ and random tests up to $N=120601$.
\end{itemize}

\section{Related Work}
Akers and Krishnamurthy introduced a group-theoretic model for symmetric interconnection networks \cite{Akers1989}. Standard references on parallel algorithms and interconnection networks discuss meshes, tori, hypercubes, collective communication, and routing primitives \cite{Leighton1992,DallyTowles2004}. Fault-tolerant communication algorithms in toroidal networks were studied in \cite{AlMohammadBose1999}.

Gaussian and Eisenstein--Jacobi networks were developed as algebraic low-degree interconnection topologies with compact routing descriptions. Martinez and coauthors modeled toroidal and hexagonal networks using Gaussian and EJ algebraic structures \cite{MartinezGaussian2008,MartinezDenseGaussian2006,MartinezEJ2008}. Flahive and Bose studied the topology of Gaussian and EJ interconnection networks \cite{FlahiveBose2010}, and later classified resource placements in Gaussian and EJ networks \cite{FlahiveBose2013}. Communication algorithms in hexagonal mesh networks and Hamiltonian structures in Gaussian networks were studied in \cite{AlbaderBoseFlahive2012,AlbaderBose2016}. Independent spanning trees in EJ networks provide multiple pre-built broadcast structures for fault-tolerant delivery \cite{HussainAboElFotohAlBdaiwi2021}. EJ-MOEM instead takes a single selected tree after the fault set is known and minimizes the number of new component-crossing edges needed to repair it.

Post-fault reconfiguration in structured networks has been studied in other topologies. Faulty hypercube broadcasting and routing algorithms use local connectivity or local safety information \cite{LeeHayes1988,LiuSong2005}. Fault-tolerant adaptive routing in dragonfly networks preserves communication under link or router failures \cite{XiangLiFu2019}.

\section{Dense EJ Network Model}
Let
\begin{equation}
\omega=\frac{-1+i\sqrt{3}}{2}.
\end{equation}
We represent an Eisenstein integer $x+y\omega$ by the axial coordinate pair $(x,y)\in\mathbb{Z}^2$. The six unit EJ directions are
\begin{equation}
\mathcal{D}=\{(1,0),(0,1),(1,-1),(-1,0),(0,-1),(-1,1)\}.
\end{equation}
For readability, these directions are denoted by
\begin{equation}
E,N,SE,W,S,NW,
\end{equation}
respectively. The EJ distance from the origin is
\begin{equation}
\rho(x,y)=d_{EJ}((0,0),(x,y))=\max\{|x|,|y|,|x+y|\}.
\end{equation}
The radius-$t$ EJ ball centered at the origin is the discrete hexagon
\begin{equation}
H_t=\{(x,y): |x|\le t,\ |y|\le t,\ |x+y|\le t\}.
\end{equation}

\begin{figure}[H]
\centering
\includegraphics[width=0.65\linewidth]{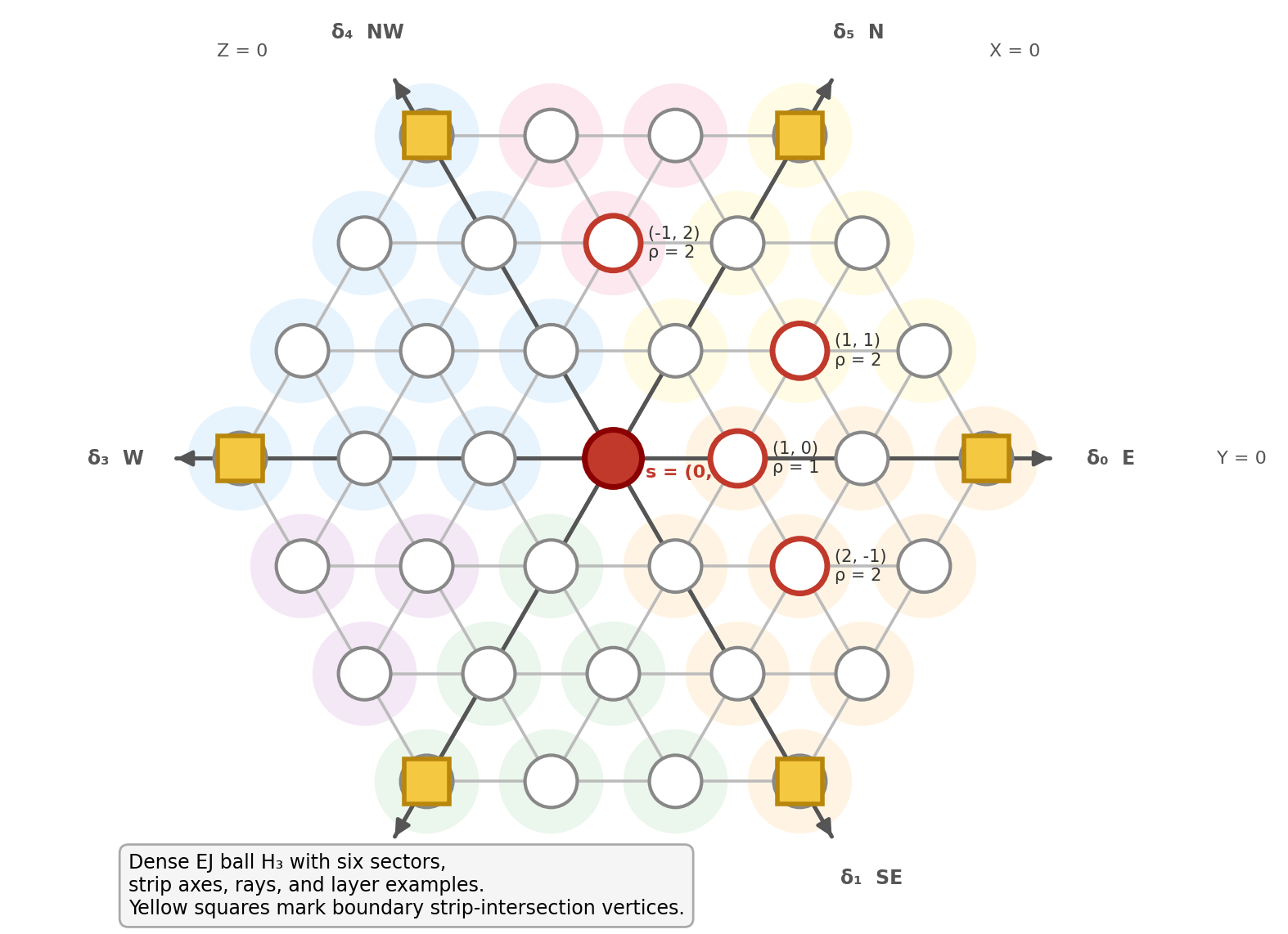}
\caption{EJ hexagonal ball $H_3$ with six sectors, strip axes, unit rays, and example layer values. Yellow squares mark boundary strip-intersection vertices where two strip coordinates are simultaneously tight.}
\label{fig:hex-sector}
\end{figure}

It contains
\begin{equation}
|H_t|=3t^2+3t+1
\end{equation}
vertices. The dense EJ network considered here is generated by
\begin{equation}
\alpha=(t+1)+t\omega,
\end{equation}
so its order is
\begin{equation}
N=3t^2+3t+1.
\end{equation}
The integer label associated with an axial coordinate is
\begin{equation}
\phi(x+y\omega)\equiv t x+(2t+1)y \pmod{N}. \label{eq:label}
\end{equation}

\begin{example}[Dense EJ coordinates for $t=3$]
For $t=3$, the network has $N=37$ nodes. The vertex $(1,0)$ has label $\phi(1,0)=3$, while $(0,1)$ has label $\phi(0,1)=7$ under the convention in \eqref{eq:label}.
\end{example}

\subsection{Fault-Free Coordinate-Reduction Broadcast Trees}
Let the source be $s=(0,0)$. A coordinate-reduction orientation chooses, for every non-source vertex $v\in H_t$, one neighbor $p(v)$ satisfying
\begin{equation}
\rho(p(v))=\rho(v)-1.
\end{equation}

\begin{lemma}[Fault-free EJ orientation tree]
For any parent rule that assigns to every non-source vertex $v\in H_t$ a neighbor $p(v)$ with $\rho(p(v))=\rho(v)-1$, the induced parent relation is a spanning tree rooted at $s$ with depth at most $t$.
\end{lemma}

\begin{proof}
Every non-source vertex has exactly one parent. Along each parent edge, the layer $\rho$ decreases by one. Therefore directed cycles are impossible and repeated parent application reaches the unique layer-zero vertex. Since every vertex of $H_t$ has layer at most $t$, the resulting tree has depth at most $t$.
\end{proof}

\begin{example}[Priority-based parent choice]
Let $t=4$ and consider $v=(2,1)$. Then $\rho(v)=\max\{2,1,3\}=3$. The neighbors $(1,1)$ and $(2,0)$ have layer $2$, so both are valid parents. An orientation priority list decides which inward neighbor is selected.
\end{example}

\subsection{Orientation Family}
The six directions are indexed in cyclic order as
\begin{align}
\delta_0&=E, & \delta_1&=SE, & \delta_2&=S,\nonumber\\
\delta_3&=W, & \delta_4&=NW, & \delta_5&=N.
\end{align}
with all indices taken modulo six.

EJ-MOEM uses a constant-size family $\Theta$ of 15 priority orientations:
\begin{align}
C_i&=(\delta_i,\delta_{i+1},\delta_{i+2},\delta_{i+3},\delta_{i+4},\delta_{i+5}), &&0\le i<6,\nonumber\\
R_i&=(\delta_i,\delta_{i-1},\delta_{i-2},\delta_{i-3},\delta_{i-4},\delta_{i-5}), &&0\le i<6,\nonumber\\
A_i&=(\delta_i,\delta_{i+3},\delta_{i+1},\delta_{i+4},\delta_{i+2},\delta_{i+5}), &&0\le i<3. \label{eq:orientationfamily}
\end{align}

\begin{table}[H]
\centering
\scriptsize
\caption{Explicit EJ-MOEM orientation family. Directions are indexed cyclically by $\delta_0=E,\delta_1=SE,\delta_2=S,\delta_3=W,\delta_4=NW,\delta_5=N$.}
\label{tab:orientations}
\begin{tabular}{ll}
\toprule
Orientation & Priority rule \\
\midrule
$C_0$ & $\delta_0,\delta_1,\delta_2,\delta_3,\delta_4,\delta_5$ \\
$C_1$ & $\delta_1,\delta_2,\delta_3,\delta_4,\delta_5,\delta_0$ \\
$C_2$ & $\delta_2,\delta_3,\delta_4,\delta_5,\delta_0,\delta_1$ \\
$C_3$ & $\delta_3,\delta_4,\delta_5,\delta_0,\delta_1,\delta_2$ \\
$C_4$ & $\delta_4,\delta_5,\delta_0,\delta_1,\delta_2,\delta_3$ \\
$C_5$ & $\delta_5,\delta_0,\delta_1,\delta_2,\delta_3,\delta_4$ \\
$R_0$ & $\delta_0,\delta_5,\delta_4,\delta_3,\delta_2,\delta_1$ \\
$R_1$ & $\delta_1,\delta_0,\delta_5,\delta_4,\delta_3,\delta_2$ \\
$R_2$ & $\delta_2,\delta_1,\delta_0,\delta_5,\delta_4,\delta_3$ \\
$R_3$ & $\delta_3,\delta_2,\delta_1,\delta_0,\delta_5,\delta_4$ \\
$R_4$ & $\delta_4,\delta_3,\delta_2,\delta_1,\delta_0,\delta_5$ \\
$R_5$ & $\delta_5,\delta_4,\delta_3,\delta_2,\delta_1,\delta_0$ \\
$A_0$ & $\delta_0,\delta_3,\delta_1,\delta_4,\delta_2,\delta_5$ \\
$A_1$ & $\delta_1,\delta_4,\delta_2,\delta_5,\delta_3,\delta_0$ \\
$A_2$ & $\delta_2,\delta_5,\delta_3,\delta_0,\delta_4,\delta_1$ \\
\bottomrule
\end{tabular}
\end{table}

\section{EJ-MOEM: Multi-Orientation Edge-Minimum Repair}
Let $T_\theta$ be the broadcast tree induced by orientation $\theta\in\Theta$. Let $F$ be a set of faulty vertices with $s\notin F$. Removing the faulty vertices and their incident tree edges gives a forest
\begin{equation}
T_\theta-F=C_1\cup C_2\cup\cdots\cup C_c.
\end{equation}
Let $C_s$ denote the component containing the source.

\begin{definition}[Component graph]
The component graph $\mathcal{C}_\theta$ has one vertex for each component $C_j$ of $T_\theta-F$. Two component vertices are adjacent in $\mathcal{C}_\theta$ if there exists a healthy EJ graph edge joining the corresponding tree components.
\end{definition}

\subsection{Certified Edge-Minimum Repair for One Orientation}
For a crossing edge $(a,b)$ from the already repaired side to an unrepaired component $C$, define the predicted attachment value
\begin{equation}
D(a,b,C)=d_{T^{\mathrm{cur}}}(s,a)+1+\Ecc_C(b), \label{eq:predicted}
\end{equation}
where $d_{T^{\mathrm{cur}}}(s,a)$ is the depth of $a$ in the current partial repaired tree and $\Ecc_C(b)$ is the maximum distance from $b$ to any vertex of $C$ using only the original fault-pruned tree edges inside $C$.

\begin{center}
\begin{tabular}{p{0.92\linewidth}}
\toprule
\textbf{Algorithm 1: Certified edge-minimum repair for one EJ orientation}\\
\midrule
Input: dense EJ network $H_t$, source $s$, fault set $F$, orientation tree $T_\theta$.\\
1. Delete $F$ and incident tree edges from $T_\theta$.\\
2. Compute components $C_1,\ldots,C_c$ and identify source component $C_s$.\\
3. Build the component graph $\mathcal{C}_\theta$ and record all healthy EJ crossing edges between components.\\
4. Enumerate the rooted component-level spanning choices of $\mathcal{C}_\theta$; for each ordered component attachment, use the stored crossing edge with minimum attachment value.\\
5. Root the resulting candidate at $s$, compute its maximum depth, and validate fault exclusion, source reachability, acyclicity, and parent uniqueness.\\
6. Return the valid candidate of minimum depth; ties prefer fewer repair edges, then lower maximum out-degree, then deterministic edge order.\\
\bottomrule
\end{tabular}
\end{center}

\subsection{Multi-Orientation Selection}
The full algorithm evaluates every orientation in $\Theta$. The best valid candidate is selected lexicographically by
\begin{multline}
(\text{success},\text{depth},\text{repair edges},\\
\text{maximum out-degree},\text{orientation rank}).
\end{multline}

\begin{center}
\begin{tabular}{p{0.92\linewidth}}
\toprule
\textbf{Algorithm 2: EJ-MOEM broadcast repair}\\
\midrule
Input: dense EJ network $H_t$, source $s$, fault set $F$, orientation family $\Theta$.\\
1. Initialize candidate set $\mathcal{B}=\emptyset$.\\
2. For every $\theta\in\Theta$, construct $T_\theta$.\\
3. Apply Algorithm 1 to $T_\theta-F$.\\
4. Insert every validated candidate into $\mathcal{B}$.\\
5. Return the best valid candidate under the lexicographic score.\\
\bottomrule
\end{tabular}
\end{center}

\section{Correctness and Edge-Minimum Repair}

\begin{lemma}[Component lower bound]
Let $T$ be a spanning tree of a graph $G$, and let $F$ be a fault set not containing the source. If $T-F$ has $c$ healthy components, then any repaired broadcast tree that preserves the healthy vertices and reconnects these components must use at least $c-1$ component-crossing edges.
\end{lemma}

\begin{proof}
Contract each connected component of $T-F$ into a supernode. A broadcast tree spanning all healthy vertices must connect these $c$ supernodes. Any connected graph on $c$ vertices has at least $c-1$ edges. Each such edge corresponds to a healthy graph edge crossing between two distinct components.
\end{proof}

\begin{theorem}[External repair-edge optimality]
Suppose the component graph $\mathcal{C}_\theta$ of $T_\theta-F$ is connected. Then Algorithm 1 returns a non-redundant broadcast tree over $H_t-F$ and uses exactly $c-1$ external component-repair edges.
\end{theorem}

\begin{proof}
Since $\mathcal{C}_\theta$ is connected, a component-level spanning tree rooted at $C_s$ exists. Algorithm 1 adds one crossing edge whenever it attaches one new component, so it adds exactly $c-1$ external edges. Inside each component, the fault-pruned structure is a tree. The resulting graph is connected, acyclic, includes every healthy vertex, and excludes every faulty vertex. The lower bound proves that $c-1$ is minimum.
\end{proof}

\begin{corollary}
EJ-MOEM returns a correct non-redundant repaired broadcast tree whenever at least one tested orientation has a connected healthy component graph.
\end{corollary}

\begin{remark}[Orientation dominance]
Since the repaired candidate for any fixed orientation $\theta\in\Theta$ is one of the candidates considered by EJ-MOEM, the selected candidate cannot have a worse score than any single fixed-orientation candidate under the declared lexicographic score.
\end{remark}

\begin{example}[Component contraction and repair]
Suppose a two-fault deletion splits an orientation tree into four healthy components: the source component and three detached suffix components. The component lower bound says that at least three external repair edges are necessary. If the component graph is connected, Algorithm 1 attaches the three detached components one at a time and uses exactly three repair edges.
\end{example}

\section{Depth Certificates in EJ Geometry}

\begin{definition}[EJ repair certificate]
Fix an orientation $\theta$ and let $T_\theta-F$ have components $C_1,\ldots,C_c$, with source component $C_s$. A $K$-depth repair certificate consists of an ordering of the non-source components and crossing edges
\begin{equation}
(a_j,b_j),\quad j=1,\ldots,c-1,
\end{equation}
where $b_j\in C_j$ and $a_j$ lies in the partial repaired tree $T^{(j-1)}$. The certificate condition is
\begin{equation}
d_{T^{(j-1)}}(s,a_j)+1+\Ecc_{C_j}(b_j)\le K. \label{eq:cert}
\end{equation}
\end{definition}

\begin{lemma}[Certificate implies bounded depth]
If an orientation $\theta$ admits a $K$-depth repair certificate, then there exists a non-redundant repaired broadcast tree of depth at most $K$.
\end{lemma}

\begin{proof}
Attach components in the certificate order. When $C_j$ is attached through $(a_j,b_j)$, every vertex $v\in C_j$ obtains a path from $s$ through $a_j$, then across $(a_j,b_j)$, and then inside $C_j$ from $b_j$ to $v$. Thus
\begin{equation}
\operatorname{dist}(s,v)\le d_{T^{(j-1)}}(s,a_j)+1+\operatorname{dist}_{C_j}(b_j,v)\le K.
\end{equation}
The final structure is connected, acyclic, includes every healthy vertex, and has maximum depth at most $K$.
\end{proof}

\begin{lemma}[Certified repair selection]
If some component-level repair of $T_\theta-F$ admits a $K$-depth repair certificate, then Algorithm 1 returns, for the same orientation, a valid candidate of depth at most $K$.
\end{lemma}

\begin{proof}
Algorithm 1 enumerates the rooted component-level spanning choices of the connected component graph. The certified repair is one of these choices. Therefore the minimum-depth candidate returned by Algorithm 1 has depth no larger than $K$.
\end{proof}

\subsection{Hexagonal Suffixes and Sector Entries}
For a vertex $v=(x,y)$, write the three strip coordinates as
\begin{equation}
X(v)=x,\quad Y(v)=y,\quad Z(v)=x+y.
\end{equation}
The layer is $\rho(v)=\max\{|X(v)|,|Y(v)|,|Z(v)|\}$.

\begin{definition}[Sector suffix]
Let $f$ be a fault of layer $r=\rho(f)$. A component $C$ of $T_\theta-\{f\}$ is an $r$-sector suffix if every vertex of $C$ has layer at least $r+1$, the parent path from every vertex in $C$ decreases the layer by one until it reaches a child of $f$, and there is an entry vertex $b\in C$ of layer $r+1$ with
\begin{equation}
\Ecc_C(b)\le t-r-1.
\end{equation}
\end{definition}

\begin{lemma}[Sector side-entry lemma]
\label{lem:sector-side-entry}
Let $C$ be an $r$-sector suffix with entry $b$. Suppose $b$ lies in the sector bounded by two adjacent rays $\delta_i$ and $\delta_{i+1}$. If the chosen priority is either $C_i$ or $R_{i+1}$, then $C$ has a side-entry crossing from outside the descendant interval. More precisely, there is a healthy neighbor $a$ of $b$ such that $a\notin C$ and either $\rho(a)=r+1$ or $\rho(a)=r$. If $a$ is not the second fault, then $a$ lies in the source component of the fault-pruned tree.
\end{lemma}

\begin{proof}
Inside the sector bounded by $\delta_i$ and $\delta_{i+1}$, write every point in sector coordinates as $u\delta_i+v\delta_{i+1}$ with $u,v\ge0$. A cyclic priority $C_i$ or the reversed priority $R_{i+1}$ chooses one of the two inward gates before the other whenever both lower the layer. The descendants of a fixed child of $f$ form a monotone interval in the first layer beyond the cut, and the two endpoints cannot both be internal descendants of the same child of $f$. Hence at least one endpoint neighbor across a sector boundary is outside $C$.

At strip-intersection vertices where two coordinates among $|X|,|Y|,|Z|$ are tight, there are two inward directions. The priority selects one as the parent direction; the alternative gate has layer $r$ and is not a descendant of $f$. The six signed strip-side cases are:
\begin{center}
\footnotesize
\begin{tabular}{@{}p{0.27\linewidth}p{0.25\linewidth}p{0.31\linewidth}@{}}
\toprule
Tight strip of $b$ & Interior range & Two layer-$r$ gates \\
\midrule
$X=r+1$ & $y<0$ & $W$, $NW$ \\
$Y=r+1$ & $x<0$ & $S$, $SE$ \\
$Z=r+1$ & $x,y>0$ & $W$, $S$ \\
$X=-(r+1)$ & $y>0$ & $E$, $SE$ \\
$Y=-(r+1)$ & $x>0$ & $N$, $NW$ \\
$Z=-(r+1)$ & $x,y<0$ & $E$, $N$ \\
\bottomrule
\end{tabular}
\end{center}
In each case the alternative gate is not a descendant of $f$, since every descendant of $f$ has layer strictly greater than $r$. The quotient identification does not change this conclusion: parent paths keep strictly decreasing canonical layer sequences and cannot pass through $f$.
\end{proof}

\begin{lemma}[Sector suffix attachment]
Let $C$ be an $r$-sector suffix with entry $b$. If there is a healthy crossing edge $(a,b)$ with $a$ in the current partial repaired tree $T^{(j-1)}$ and $d_{T^{(j-1)}}(s,a)\le r+1$, then attaching $C$ through $(a,b)$ has certificate value at most $t+1$. If $d_{T^{(j-1)}}(s,a)\le r+2$, the value is at most $t+2$.
\end{lemma}

\begin{proof}
By the suffix definition, $\Ecc_C(b)\le t-r-1$. Therefore
\begin{equation}
d_{T^{(j-1)}}(s,a)+1+\Ecc_C(b)\le (r+1)+1+(t-r-1)=t+1.
\end{equation}
The same substitution with $d_{T^{(j-1)}}(s,a)\le r+2$ gives $t+2$.
\end{proof}

\begin{theorem}[One-fault EJ depth theorem]
For every $t\ge1$ and every one-fault placement $F=\{f\}\subseteq H_t\setminus\{s\}$, the EJ-MOEM orientation family contains a repaired tree of depth at most $t+1$.
\end{theorem}

\begin{proof}
Let $r=\rho(f)$. Choose a cyclic or reversed cyclic orientation whose first two priorities match the sector containing the outward children of $f$. Each detached component is an $r$-sector suffix. By the sector side-entry lemma, each suffix has a healthy lateral entry from the source component. The repaired-side endpoint has depth at most $r+1$ in the original tree. The sector suffix attachment lemma gives value at most $t+1$ for every detached component.
\end{proof}

\begin{remark}[Explicit one-fault witness]
If the failed vertex $f$ lies at layer $r$ in sector $S_i$, take the cyclic or reverse orientation whose first priorities expose the two boundary directions of $S_i$. Let $b$ be the first layer-$(r+1)$ vertex of the descendant suffix cut off by $f$. The repair edge is the lateral edge $(a,b)$ where $a$ is the side-entry neighbor supplied by the sector side-entry lemma. Figure~\ref{fig:one-fault-witness} shows the concrete witness used in Example~4.
\end{remark}

\begin{figure}[H]
\centering
\includegraphics[width=0.65\linewidth]{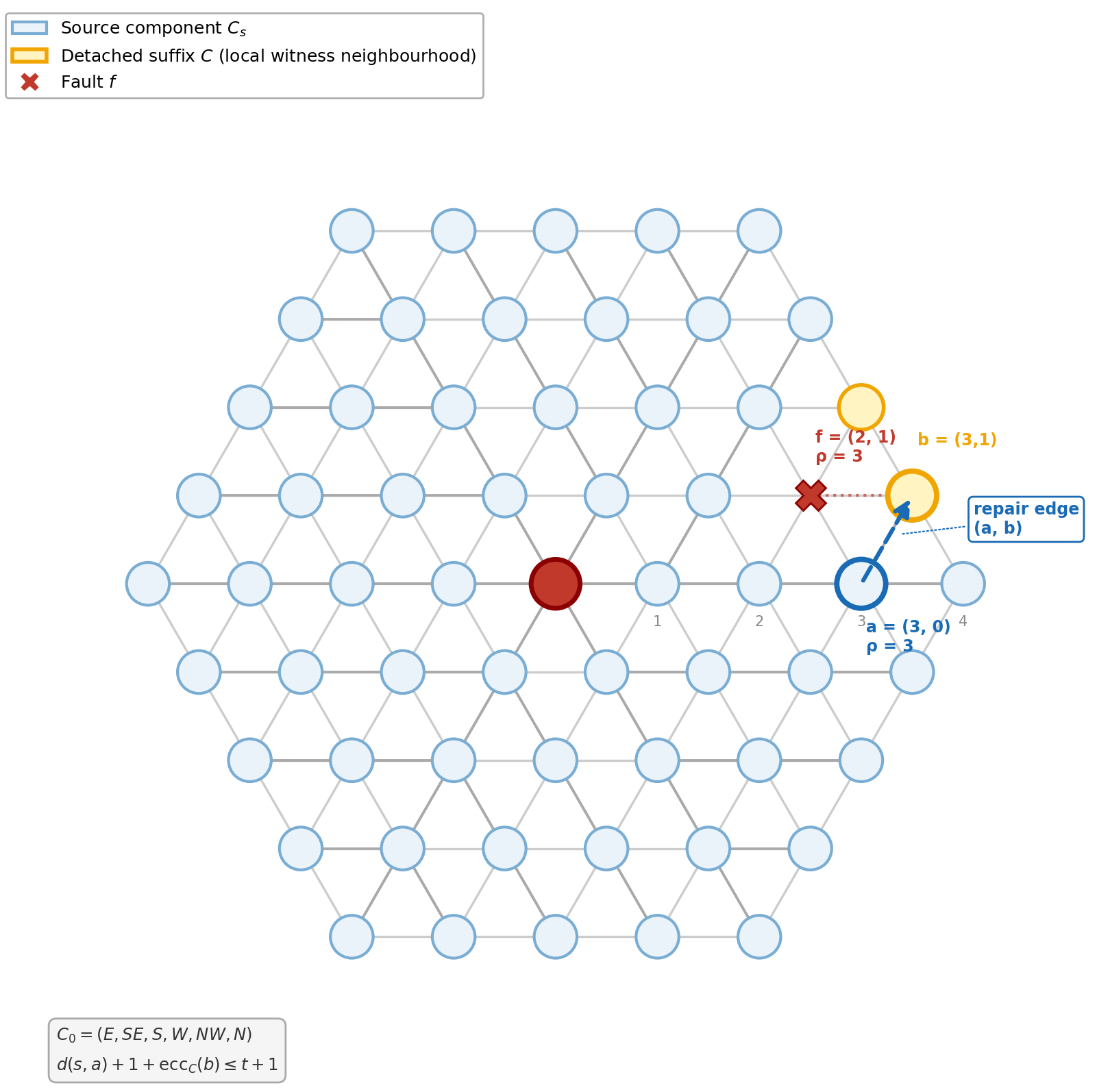}
\caption{One-fault side-entry witness for $t=4$. The fault $f=(2,1)$ cuts off a local sector suffix; the dashed repair edge $((3,0),(3,1))$ reconnects the suffix from the source component.}
\label{fig:one-fault-witness}
\end{figure}

\subsection{Two-Fault Certificate Classification}

\begin{lemma}[Non-adjacent sector separation]
\label{lem:nonadjacent-separation}
Let $f$ and $g$ be two faults in non-adjacent sectors of $H_t$, and suppose that they are not on opposite rays. Then the side-entry endpoint $a$ supplied for a suffix detached by $f$ is neither $g$ nor a vertex blocked solely by $g$.
\end{lemma}

\begin{proof}
Each side-entry edge changes one of the three strip coordinates $X,Y,Z$ by $\pm1$. A fault in a non-adjacent, non-opposite sector is separated from both side-entry boundaries by at least one full sector, so it cannot be the side-entry endpoint and cannot occupy the crossing edge used for that suffix. The descendant cones are disjoint except at opposite-ray boundary cases.
\end{proof}

\begin{lemma}[Non-adjacent paired cuts]
\label{lem:nonadjacent-components}
For two non-adjacent, non-opposite-ray faults, a side-entry endpoint for one suffix that is not in the source component can only lie in the component detached by the other cut.
\end{lemma}

\begin{proof}
The outward cone of the second fault's descendants is separated from the first suffix's boundary lines by at least one full sector. Therefore deleting the other fault can block the parent path from the side-entry endpoint only if that endpoint lies on the boundary of the other fault's own detached cone. If it were in a third component, the unique parent path would have to cross two distinct sector boundaries without passing through either deleted vertex, contradicting the coordinate-reduction parent rule.
\end{proof}

\begin{lemma}[Same-ray and opposite-ray faults]
\label{lem:same-opposite}
If two faults lie on the same EJ ray or on two opposite EJ rays, then some orientation in $\Theta$ admits a $(t+2)$-depth repair certificate.
\end{lemma}

\begin{proof}
By symmetry take the ray to be $\delta_0=E$. For same-ray faults at layers $1\le a<b\le t$, use the alternating orientation $A_0=(E,W,SE,NW,S,N)$. The bounded interval between the cuts is entered from a transverse neighbor at depth at most $a+1$ and has internal eccentricity at most $b-a-1$, giving value at most $b+1\le t+1$. The outer suffix is entered at depth at most $b+1$ and has eccentricity at most $t-b$, giving value at most $t+2$.

For opposite-ray faults, the same opposite-pair priority separates the two tails. Each tail has a transverse entry with value at most $t+2$.
\end{proof}

\begin{lemma}[Separated-sector faults]
\label{lem:separated}
If two faults lie in non-adjacent sectors and are not on opposite rays, then some orientation in $\Theta$ admits a $(t+2)$-depth repair certificate.
\end{lemma}

\begin{proof}
Choose the cyclic or reverse orientation exposing the boundary farther from the second fault. By Lemma~\ref{lem:nonadjacent-separation}, the side-entry neighbor for the first suffix cannot be the second fault. For the second suffix, by Lemma~\ref{lem:nonadjacent-components}, the only possible obstruction is the other cut's component, which is attached first. The first attachment has value at most $t+1$; after it is attached, the second has value at most $t+2$.
\end{proof}

\begin{lemma}[Adjacent-sector faults]
\label{lem:adjacent}
If two faults lie in adjacent sectors sharing a boundary ray, then some orientation in $\Theta$ admits a $(t+2)$-depth repair certificate.
\end{lemma}

\begin{proof}
Normalize the common boundary to be $\delta_0=E$. Use orientation $A_0=(E,W,SE,NW,S,N)$. If both side entries are healthy, the sector suffix attachment lemma gives value at most $t+1$ for both. If one entry is the other fault, attach the component with the unblocked outer side first with value at most $t+1$; after it is attached, the second component has value at most $t+2$.
\end{proof}

\begin{lemma}[Boundary microcomponents]
\label{lem:boundary-micro}
If a detached component is represented only through a quotient wraparound edge at the boundary of $H_t$, then it still has a $(t+2)$ attachment.
\end{lemma}

\begin{proof}
A boundary microcomponent has one or two consecutive boundary vertices, so $\Ecc_C(b)\le1$. The repaired-side endpoint $a$ is a vertex of the original source component with depth at most $t$. If $C$ is a single boundary vertex, the certificate value is at most $t+1$. If $C$ has two consecutive boundary vertices, the value is at most $t+2$.
\end{proof}

\begin{theorem}[EJ-MOEM $t+2$ depth theorem]
\label{thm:depth}
For every dense EJ network $H_t$ with $t\ge2$ and every fault set $F\subseteq H_t\setminus\{s\}$ with $|F|\le2$, the EJ-MOEM orientation family contains a valid repair of depth at most $t+2$. If $|F|=1$, depth at most $t+1$ is sufficient.
\end{theorem}

\begin{proof}
The case $|F|=0$ is the fault-free tree of depth at most $t$. The case $|F|=1$ is the one-fault depth theorem. For $|F|=2$: ray-resident fault pairs are assigned to Lemma~\ref{lem:same-opposite}; non-adjacent, non-opposite-ray sector pairs are assigned to Lemma~\ref{lem:separated}; adjacent-sector pairs are assigned to Lemma~\ref{lem:adjacent}; boundary quotient representatives are covered by Lemma~\ref{lem:boundary-micro}. These cases are mutually exclusive under the stated priority order and exhaustive after dihedral normalization. Hence some orientation in $\Theta$ admits a $(t+2)$ certificate, the component graph is connected, and Algorithm 1 returns a candidate of depth at most $t+2$.
\end{proof}

\begin{corollary}[Minimum repair count under the depth bound]
For every one- or two-fault placement covered by the preceding theorem, the returned EJ-MOEM tree satisfies the stated depth bound and uses exactly $c-1$ external component-repair edges for its selected fault-pruned orientation.
\end{corollary}

\begin{table}[H]
\centering
\caption{EJ depth-certificate templates used in the $t+2$ proof.}
\label{tab:templates}
\begin{adjustbox}{max width=\textwidth}
\begin{tabular}{p{0.22\linewidth}p{0.25\linewidth}p{0.25\linewidth}p{0.18\linewidth}}
\toprule
Fault signature & Orientation strategy & Detached component shape & Certificate bound \\
\midrule
One interior-sector fault & Choose $C_i$ or $R_i$ exposing a side entry & One monotone hexagonal suffix & $t+1$ \\
One boundary fault & Use transverse priority & One boundary arc or small cap & $t+1$ \\
Two faults on same ray & Use $A_i$ to expose transverse directions & One or two nested axial tails & $t+2$ \\
Two faults on opposite rays & Use $A_i$ to separate the tails & Two independent tails & $t+2$ \\
Two faults in separated sectors & Use cyclic or reverse priority & Two sector suffixes & $t+2$ \\
Two faults in adjacent sectors & Use $A_i$ around shared boundary ray & Two sector suffixes & $t+2$ \\
Boundary wraparound pair & Use quotient representative & One- or two-node microcomponent & $t+2$ \\
\bottomrule
\end{tabular}
\end{adjustbox}
\end{table}

\begin{table}[H]
\caption{Normalized two-fault sector-case audit.}
\label{tab:case-audit}
\centering
\begin{adjustbox}{max width=\textwidth}
\begin{tabular}{p{0.26\linewidth}p{0.35\linewidth}p{0.26\linewidth}}
\toprule
Class & Normalized condition & Covering result \\
\midrule
Ray-resident pair & Both faults on same ray or opposite rays & Lemma~\ref{lem:same-opposite} \\
Separated sectors & Sector indices differ by two, three, or four modulo six & Lemmas~\ref{lem:nonadjacent-components} and~\ref{lem:separated} \\
Adjacent sectors & Sectors share one boundary ray & Lemma~\ref{lem:adjacent} \\
Boundary representative & Side-entry edge via quotient boundary & Lemma~\ref{lem:boundary-micro} \\
\bottomrule
\end{tabular}
\end{adjustbox}
\end{table}

\begin{example}[One-fault depth certificate]
For $t=4$, let $f=(2,1)$ with $\rho(f)=3$. Choose orientation $C_0=(E,SE,S,W,NW,N)$. The first layer-4 entry in the descendant suffix is $b=(3,1)$; its side-entry neighbor is $a=(3,0)$, which has layer $3$ and lies in the source component. The repair edge is $((3,0),(3,1))$ with certificate value $3+1+0=4=t$.
\end{example}

\begin{example}[Two-fault depth certificate]
For $t=4$, suppose two faults lie in adjacent sectors: $f_1=(1,0)$ and $f_2=(0,1)$. The alternating opposite-pair orientation attaches the component with the unblocked outer side first with value at most $t+1$. After it is attached, the second component attaches with value at most $t+2$.
\end{example}

\begin{figure}[H]
\centering
\includegraphics[width=0.65\linewidth]{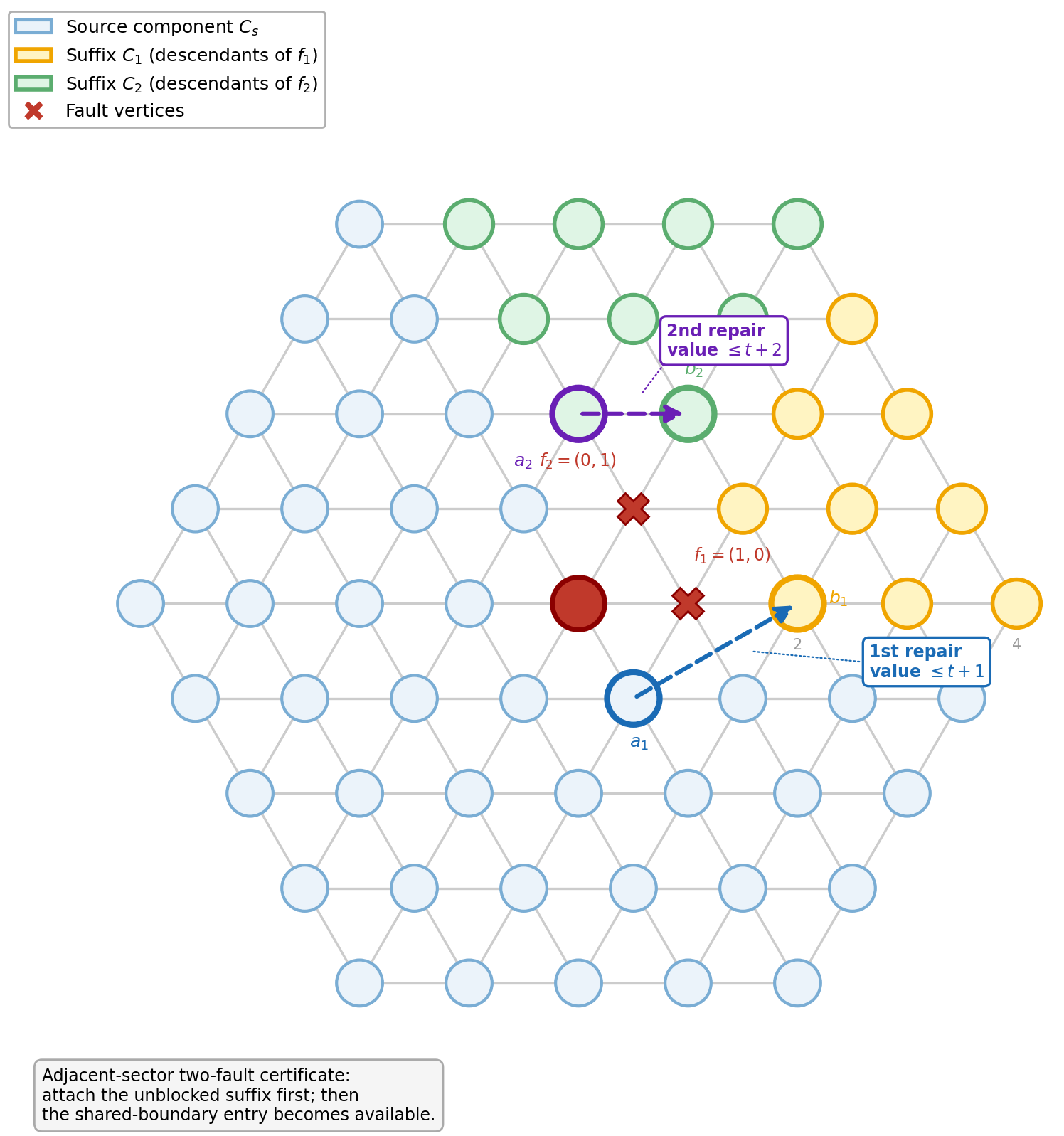}
\caption{Schematic adjacent-sector two-fault repair for $t=4$. After the unblocked suffix is attached, the shared-boundary entry for the second suffix becomes available.}
\label{fig:two-fault-schematic}
\end{figure}

\begin{example}[Complete two-fault repair for $t=3$]
Let $t=3$, $N=37$, $s=(0,0)$, and fault set $F=\{(1,0),(2,0)\}$. Use the alternating orientation $A_0=(E,W,SE,NW,S,N)$. After deleting the two faults, the only outer suffix beyond $(2,0)$ is $C_2=\{(3,0)\}$. The neighbors of $(3,0)$ include the healthy boundary vertex $(2,1)\in C_s$. Choose crossing edge $((2,1),(3,0))$. Since $\Ecc_{C_2}((3,0))=0$ and $d_{C_s}(s,(2,1))=3$, the certificate value is $3+1+0=4=t+1$. The repaired tree uses exactly one repair edge.
\end{example}

\section{Experimental Validation}
The validation suite checks the implementation against the proved invariants. Every reported case verifies fault exclusion, source reachability, parent uniqueness, acyclicity, number of external repair edges, number of components, final depth, and the selected orientation.

\subsection{Extended Exhaustive Validation}

\begin{table}[H]
\centering
\scriptsize
\caption{Extended exhaustive EJ-MOEM validation. Every one- and two-fault placement was tested for the listed diameters.}
\label{tab:exhaustive}
\begin{adjustbox}{max width=\textwidth}
\begin{tabular}{rrrrrrrrrr}
\toprule
$t$ & $N$ & $|F|$ & Cases & Success & Max depth & Max over. & Max repair & Max comp. & Min good \\
\midrule
2 & 19 & 1 & 18 & 100\% & 3 & 1 & 1 & 2 & 15 \\
2 & 19 & 2 & 153 & 100\% & 3 & 1 & 3 & 4 & 15 \\
3 & 37 & 1 & 36 & 100\% & 4 & 1 & 1 & 2 & 15 \\
3 & 37 & 2 & 630 & 100\% & 5 & 2 & 3 & 4 & 15 \\
4 & 61 & 1 & 60 & 100\% & 5 & 1 & 1 & 2 & 12 \\
4 & 61 & 2 & 1,770 & 100\% & 5 & 1 & 4 & 5 & 7 \\
5 & 91 & 1 & 90 & 100\% & 6 & 1 & 1 & 2 & 12 \\
5 & 91 & 2 & 4,005 & 100\% & 6 & 1 & 4 & 5 & 7 \\
6 & 127 & 1 & 126 & 100\% & 7 & 1 & 1 & 2 & 12 \\
6 & 127 & 2 & 7,875 & 100\% & 7 & 1 & 4 & 5 & 7 \\
7 & 169 & 1 & 168 & 100\% & 8 & 1 & 1 & 2 & 12 \\
7 & 169 & 2 & 14,028 & 100\% & 8 & 1 & 4 & 5 & 7 \\
8 & 217 & 1 & 216 & 100\% & 9 & 1 & 1 & 2 & 12 \\
8 & 217 & 2 & 23,220 & 100\% & 9 & 1 & 4 & 5 & 7 \\
9 & 271 & 1 & 270 & 100\% & 10 & 1 & 1 & 2 & 12 \\
9 & 271 & 2 & 36,315 & 100\% & 10 & 1 & 4 & 5 & 7 \\
10 & 331 & 1 & 330 & 100\% & 11 & 1 & 1 & 2 & 12 \\
10 & 331 & 2 & 54,285 & 100\% & 11 & 1 & 4 & 5 & 7 \\
11 & 397 & 1 & 396 & 100\% & 12 & 1 & 1 & 2 & 12 \\
11 & 397 & 2 & 78,210 & 100\% & 12 & 1 & 4 & 5 & 7 \\
12 & 469 & 1 & 468 & 100\% & 13 & 1 & 1 & 2 & 12 \\
12 & 469 & 2 & 109,278 & 100\% & 13 & 1 & 4 & 5 & 7 \\
14 & 631 & 1 & 630 & 100\% & 15 & 1 & 1 & 2 & 12 \\
14 & 631 & 2 & 198,135 & 100\% & 15 & 1 & 4 & 5 & 7 \\
16 & 817 & 1 & 816 & 100\% & 17 & 1 & 1 & 2 & 12 \\
16 & 817 & 2 & 332,520 & 100\% & 17 & 1 & 4 & 5 & 7 \\
18 & 1027 & 1 & 1,026 & 100\% & 19 & 1 & 1 & 2 & 12 \\
18 & 1027 & 2 & 525,825 & 100\% & 19 & 1 & 4 & 5 & 7 \\
\bottomrule
\end{tabular}
\end{adjustbox}
\end{table}

\subsection{Structured and Large-Scale Validation}

\begin{table}[H]
\centering
\scriptsize
\caption{Structured theorem-critical and sampled stress tests. All rows use two faults.}
\label{tab:structured}
\begin{adjustbox}{max width=\textwidth}
\begin{tabular}{rrlrrrrrr}
\toprule
$t$ & $N$ & Mode & Cases/trials & Success & Avg over. & Max over. & Avg repair & Max repair \\
\midrule
15 & 721 & same-ray & 630 & 100\% & 0.867 & 1 & 1.733 & 2 \\
15 & 721 & opposite-ray & 675 & 100\% & 0.996 & 1 & 1.867 & 2 \\
15 & 721 & adjacent near-source & 42 & 100\% & 1.000 & 1 & 2.143 & 3 \\
15 & 721 & adjacent sector grid & 216 & 100\% & 0.000 & 0 & 1.667 & 2 \\
15 & 721 & boundary micro & 702 & 100\% & 0.171 & 1 & 0.752 & 2 \\
15 & 721 & non-adjacent sector & 324 & 100\% & 0.000 & 0 & 1.667 & 2 \\
15 & 721 & uniform random & 5,000 & 100\% & 0.227 & 1 & 1.767 & 3 \\
20 & 1261 & same-ray & 1,140 & 100\% & 0.900 & 1 & 1.800 & 2 \\
20 & 1261 & opposite-ray & 1,200 & 100\% & 0.998 & 1 & 1.900 & 2 \\
20 & 1261 & adjacent near-source & 42 & 100\% & 1.000 & 1 & 2.143 & 3 \\
20 & 1261 & adjacent sector grid & 216 & 100\% & 0.000 & 0 & 1.667 & 2 \\
20 & 1261 & boundary micro & 942 & 100\% & 0.159 & 1 & 0.752 & 2 \\
20 & 1261 & non-adjacent sector & 324 & 100\% & 0.000 & 0 & 1.667 & 2 \\
20 & 1261 & uniform random & 5,000 & 100\% & 0.164 & 1 & 1.807 & 3 \\
25 & 1951 & same-ray & 1,800 & 100\% & 0.920 & 1 & 1.840 & 2 \\
25 & 1951 & opposite-ray & 1,875 & 100\% & 0.998 & 1 & 1.920 & 2 \\
25 & 1951 & adjacent near-source & 42 & 100\% & 1.000 & 1 & 2.143 & 3 \\
25 & 1951 & adjacent sector grid & 216 & 100\% & 0.000 & 0 & 1.667 & 2 \\
25 & 1951 & boundary micro & 1,182 & 100\% & 0.152 & 1 & 0.751 & 2 \\
25 & 1951 & non-adjacent sector & 324 & 100\% & 0.000 & 0 & 1.667 & 2 \\
25 & 1951 & uniform random & 5,000 & 100\% & 0.141 & 1 & 1.839 & 3 \\
30 & 2791 & same-ray & 2,610 & 100\% & 0.933 & 1 & 1.867 & 2 \\
30 & 2791 & opposite-ray & 2,700 & 100\% & 0.999 & 1 & 1.933 & 2 \\
30 & 2791 & adjacent near-source & 42 & 100\% & 1.000 & 1 & 2.143 & 3 \\
30 & 2791 & adjacent sector grid & 216 & 100\% & 0.000 & 0 & 1.667 & 2 \\
30 & 2791 & boundary micro & 1,422 & 100\% & 0.148 & 1 & 0.751 & 2 \\
30 & 2791 & non-adjacent sector & 324 & 100\% & 0.000 & 0 & 1.667 & 2 \\
30 & 2791 & uniform random & 5,000 & 100\% & 0.127 & 1 & 1.869 & 3 \\
\bottomrule
\end{tabular}
\end{adjustbox}
\end{table}

\begin{table}[H]
\centering
\scriptsize
\caption{Large random validation. Time is wall-clock seconds for the complete row.}
\label{tab:random}
\begin{adjustbox}{max width=\textwidth}
\begin{tabular}{rrrrrrrrrrrr}
\toprule
$t$ & $N$ & $|F|$ & Trials & Success & Avg depth & Max depth & Avg over. & Max over. & Avg repair & Max repair & Time (s) \\
\midrule
50 & 7651 & 1 & 1,000 & 100\% & 50.033 & 51 & 0.033 & 1 & 0.952 & 1 & 87.2 \\
50 & 7651 & 2 & 1,000 & 100\% & 50.084 & 51 & 0.084 & 1 & 1.910 & 2 & 122.1 \\
100 & 30301 & 1 & 1,000 & 100\% & 100.020 & 101 & 0.020 & 1 & 0.982 & 1 & 355.4 \\
100 & 30301 & 2 & 1,000 & 100\% & 100.036 & 101 & 0.036 & 1 & 1.961 & 2 & 487.5 \\
200 & 120601 & 1 & 1,000 & 100\% & 200.014 & 201 & 0.014 & 1 & 0.986 & 1 & 1477.7 \\
200 & 120601 & 2 & 1,000 & 100\% & 200.019 & 201 & 0.019 & 1 & 1.980 & 2 & 1981.8 \\
\bottomrule
\end{tabular}
\end{adjustbox}
\end{table}

\begin{figure}[H]
\centering
\includegraphics[width=0.75\linewidth]{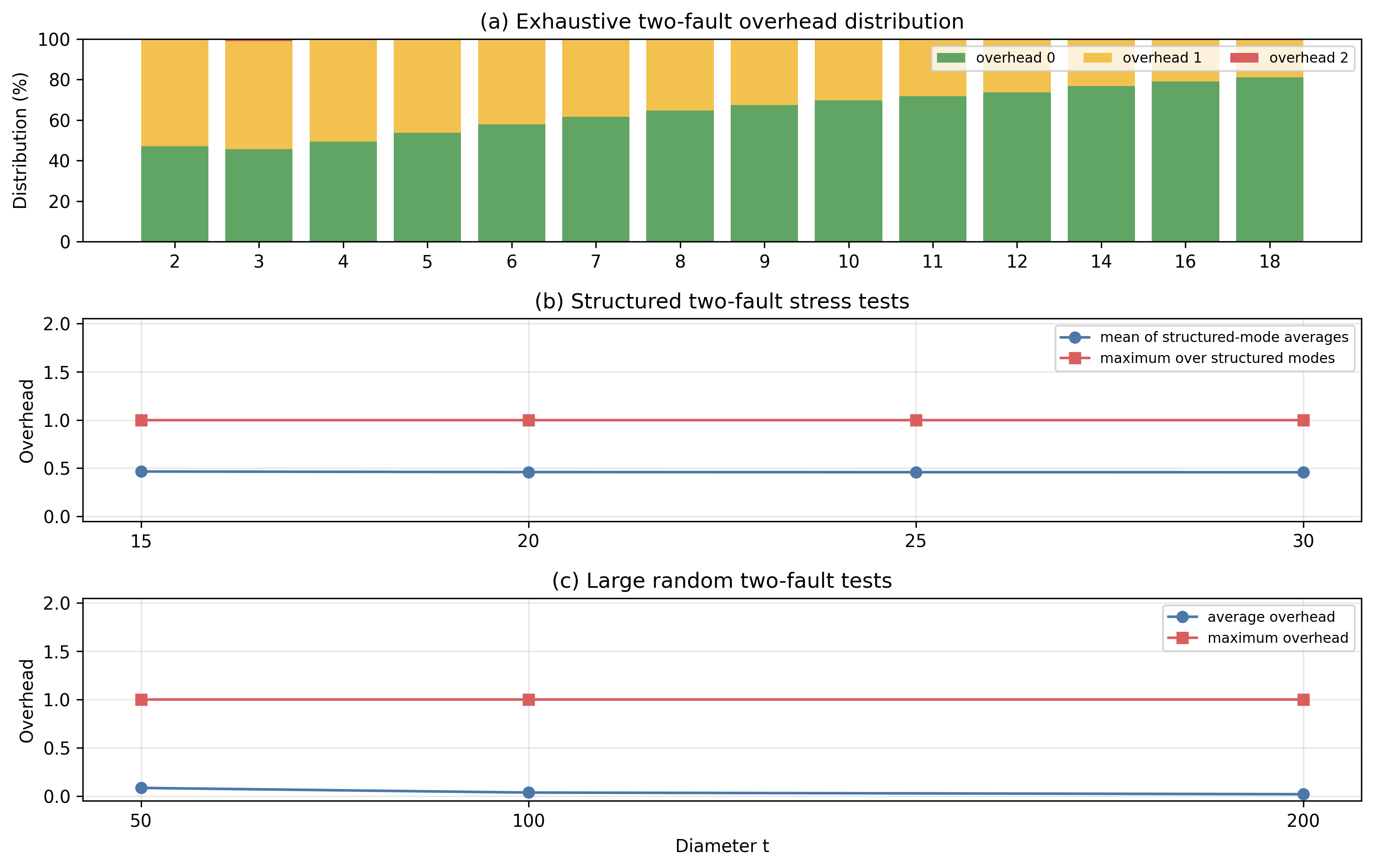}
\caption{Depth-overhead summary for the validation run. Panel (a) gives percentage distributions of exact two-fault overheads from exhaustive enumeration; panels (b) and (c) summarize average and maximum overhead for structured stress tests and large random tests, respectively.}
\label{fig:overhead-summary}
\end{figure}

\subsection{Interpretation}
The validation results support three points. First, the implementation realizes the proved depth guarantees on exhaustive and stress-tested finite instances. Second, the certified component-repair procedure realizes the predicted $c-1$ edge count. Third, although the proof guarantees depth $t+2$ for all two-fault placements, overhead two appears only in the small $t=3$ exhaustive case, while all larger exhaustive, structured, and random tests observed maximum overhead at most one.

\subsection{Baseline Interpretation}

\begin{table}[H]
\centering
\scriptsize
\caption{Objective-level and numerical comparison with natural baselines. Numerical rows: $t=10$, $N=331$, $F=\{(5,0),(10,0)\}$, $s=(0,0)$.}
\label{tab:baseline}
\begin{adjustbox}{max width=\textwidth}
\begin{tabular}{p{0.19\linewidth}p{0.22\linewidth}c c p{0.25\linewidth}}
\toprule
Method & Optimizes & New parent edges & Depth & Limitation \\
\midrule
Global BFS rebuild & Reachability & Up to $\Theta(N)=330$ & $t=10$ & May replace the damaged tree entirely \\
Fixed orientation $C_0$ & Component repair & $c-1=2$ & $11=t+1$ & Preselected; may not be optimal \\
Fixed orientation $A_0$ & Component repair & $c-1=1$ & $10=t$ & Not known before faults observed \\
Independent-tree redundancy & Path/tree diversity & Not applicable & Not comparable & Precomputes trees; does not minimize repair edges \\
EJ-MOEM & Best validated orientation & $c-1=1$ & $10=t$ & Depth proof covers $|F|\le2$ \\
\bottomrule
\end{tabular}
\end{adjustbox}
\end{table}

\section{Complexity}
For one orientation, constructing the coordinate-reduction tree takes $O(N)$ time. Deleting faults and computing connected components takes $O(N)$ time. Crossing edges are found by scanning the six neighbors of each healthy vertex, also $O(N)$. Component eccentricities are computed by two tree sweeps inside each component in linear time. EJ-MOEM evaluates a constant number of orientations, so the total running time is $O(N)$ for the one- and two-fault setting.

\section{Discussion}
The EJ case is not merely the Gaussian case with six directions. Gaussian coordinate balls are diamonds under the Manhattan metric, while EJ coordinate balls are hexagons described by three axial strips. At strip-intersection vertices, an EJ boundary or sector-corner cut may have two simultaneous tight strip coordinates, so one side entry is a layer-$(r+1)$ lateral vertex while the other is a layer-$r$ inward gate. The adjacent-sector pair case also has no direct degree-four Gaussian analogue.

The $t+2$ bound is a worst-case theorem. The validation run shows that the bound is attained in the smallest nontrivial two-fault case $t=3$, but no overhead-two case was observed for $t\ge4$ in the extended exhaustive run. A natural open question is whether depth $t+1$ is universally sufficient for all $t\ge4$.

\section{Conclusion}
This paper introduced EJ-MOEM, a multi-orientation edge-minimum repair method for non-redundant one-to-all broadcasting in dense Eisenstein--Jacobi networks. EJ-MOEM evaluates a constant-size family of six-direction coordinate-reduction broadcast trees, deletes faulty vertices, contracts the fault-pruned forest into healthy components, and reconnects those components using an externally edge-minimum set of repair edges. For any selected orientation with $c$ healthy components and connected component graph, exactly $c-1$ external component-crossing repair edges are necessary and sufficient. The depth theorem proves repairs of depth at most $t+1$ for one fault and $t+2$ for two faults. Extended validation confirms the proof-derived behavior with 100\% success across all tested cases.

\section*{Data and Code Availability}
The data and code that support the findings of this study are available from the corresponding author upon reasonable request.

\section*{Acknowledgment}
The author thanks the Department of Computer Science, Faculty of Science, Kuwait University, for its support and research environment. This work did not receive a specific grant from any funding agency in the public, commercial, or not-for-profit sectors.


\begin{thebibliography}{99}

\bibitem{Akers1989}
S. B. Akers and B. Krishnamurthy, ``A group-theoretic model for symmetric interconnection networks,'' \emph{IEEE Transactions on Computers}, vol. 38, no. 4, pp. 555--566, Apr. 1989.

\bibitem{Leighton1992}
F. T. Leighton, \emph{Introduction to Parallel Algorithms and Architectures: Arrays, Trees, Hypercubes}. San Mateo, CA, USA: Morgan Kaufmann, 1992.

\bibitem{DallyTowles2004}
W. J. Dally and B. Towles, \emph{Principles and Practices of Interconnection Networks}. San Francisco, CA, USA: Morgan Kaufmann, 2004.

\bibitem{AlMohammadBose1999}
B. F. A. AlMohammad and B. Bose, ``Fault-tolerant communication algorithms in toroidal networks,'' \emph{IEEE Transactions on Parallel and Distributed Systems}, vol. 10, no. 10, pp. 976--983, Oct. 1999.

\bibitem{LeeHayes1988}
T. C. Lee and J. P. Hayes, ``Routing and broadcasting in faulty hypercube computers,'' in \emph{Proc. Third Conference on Hypercube Concurrent Computers and Applications}, vol. 1, 1988, pp. 346--354.

\bibitem{LiuSong2005}
F. Liu and Y. Song, ``Broadcast in the locally $k$-subcube-connected hypercube networks with faulty tolerance,'' in \emph{Networking and Mobile Computing}, Lecture Notes in Computer Science, vol. 3619. Berlin, Germany: Springer, 2005, pp. 305--313.

\bibitem{XiangLiFu2019}
D. Xiang, B. Li, and Y. Fu, ``Fault-tolerant adaptive routing in dragonfly networks,'' \emph{IEEE Transactions on Dependable and Secure Computing}, vol. 16, no. 2, pp. 259--271, Mar./Apr. 2019.

\bibitem{HussainAboElFotohAlBdaiwi2021}
Z. Hussain, H. AboElFotoh, and B. AlBdaiwi, ``Independent spanning trees in Eisenstein--Jacobi networks,'' arXiv:2101.09797 [cs.DC], Jan. 2021.

\bibitem{MartinezGaussian2008}
C. Martinez, R. Beivide, E. Stafford, M. Moreto, and E. M. Gabidulin, ``Modeling toroidal networks with the Gaussian integers,'' \emph{IEEE Transactions on Computers}, vol. 57, no. 8, pp. 1046--1056, Aug. 2008.

\bibitem{MartinezDenseGaussian2006}
C. Martinez, E. Vallejo, R. Beivide, C. Izu, and M. Moreto, ``Dense Gaussian networks: Suitable topologies for on-chip multiprocessors,'' \emph{International Journal of Parallel Programming}, vol. 34, no. 3, pp. 193--211, 2006.

\bibitem{MartinezEJ2008}
C. Martinez, E. Stafford, R. Beivide, and E. M. Gabidulin, ``Modeling hexagonal constellations with Eisenstein--Jacobi graphs,'' \emph{Problems of Information Transmission}, vol. 44, no. 1, pp. 1--11, 2008.

\bibitem{FlahiveBose2010}
M. Flahive and B. Bose, ``The topology of Gaussian and Eisenstein--Jacobi interconnection networks,'' \emph{IEEE Transactions on Parallel and Distributed Systems}, vol. 21, no. 8, pp. 1132--1142, Aug. 2010.

\bibitem{AlbaderBoseFlahive2012}
B. A. Albader, B. Bose, and M. Flahive, ``Efficient communication algorithms in hexagonal mesh interconnection networks,'' \emph{IEEE Transactions on Parallel and Distributed Systems}, vol. 23, no. 1, pp. 69--77, Jan. 2012.

\bibitem{FlahiveBose2013}
M. Flahive and B. Bose, ``On resource placements in Gaussian and EJ networks,'' \emph{IEEE Transactions on Computers}, vol. 62, no. 3, pp. 623--626, Mar. 2013.

\bibitem{AlbaderBose2016}
B. A. Albader and B. Bose, ``Edge-disjoint Hamiltonian cycles in Gaussian networks,'' \emph{IEEE Transactions on Computers}, vol. 65, no. 1, pp. 315--321, Jan. 2016.

\bibitem{Camarero2013}
C. Camarero, C. Martinez, and R. Beivide, ``Symmetric interconnection networks from cubic crystal lattices,'' arXiv:1311.2019, 2013.

\bibitem{Camarero2015}
C. Camarero, C. Martinez, E. Vallejo, and R. Beivide, ``Projective networks: Topologies for large parallel computer systems,'' arXiv:1512.07574, 2015.

\bibitem{Kliazovich2013}
D. Kliazovich, F. Granelli, and D. Miorandi, ``A survey of fault-tolerant network-on-chip architectures,'' \emph{IEEE Communications Surveys and Tutorials}, vol. 15, no. 4, pp. 1676--1690, 2013.

\bibitem{PasrichaDutt2008}
S. Pasricha and N. Dutt, \emph{On-Chip Communication Architectures: System on Chip Interconnect}. San Francisco, CA, USA: Morgan Kaufmann, 2008.

\end{thebibliography}
\end{document}